\documentclass{article}

\usepackage{amssymb}
\usepackage[sort,numbers]{natbib}
\usepackage{authblk}
\usepackage{a4wide}
\usepackage{times}

\usepackage{soul}
\usepackage{url}
\usepackage[hidelinks]{hyperref}
\usepackage[utf8]{inputenc}
\usepackage[small]{caption}
\usepackage{graphicx}
\usepackage{tabularx}
\usepackage{colortbl}
\usepackage{paralist}
\usepackage{amsmath}
\usepackage{todonotes}
\usepackage{amsthm}
\usepackage{amssymb}
\usepackage{amsfonts}
\usepackage{mathtools}
\usepackage{booktabs}
\usepackage{algorithm}
\usepackage{algorithmic}
\usepackage[capitalize]{cleveref}
\usepackage{tikz}
\usetikzlibrary{arrows}
\usetikzlibrary{shapes}
\usetikzlibrary{decorations.pathreplacing}
\usetikzlibrary{shapes.geometric}
\usetikzlibrary{patterns}
\usetikzlibrary{fadings}

\usepackage{thm-restate}

\definecolor{ourblue}{RGB}{130,200,250}
\definecolor{ourdarkblue}{RGB}{30,60,140}
\definecolor{ourgreen}{RGB}{0,100,0}
\definecolor{ourred}{RGB}{180,20,30}
\definecolor{lipicsyellow}{RGB}{220,200,0}
\definecolor{ourorange}{RGB}{230,100,0}
\definecolor{ourgray}{RGB}{100,100,100}

\tikzstyle{vert2}=[circle,inner sep=1.5,fill=white,draw,minimum size=.2cm]
\tikzstyle{vert3}=[inner sep=1.5,fill=white,draw=black,minimum size=.2cm]
\usetikzlibrary{calc}
\usetikzlibrary{arrows.meta}

\usepackage{caption} %
\usepackage{subcaption}

\newtheorem{theorem}{Theorem}

\newtheorem{observation}[theorem]{Observation}

\theoremstyle{definition}

\newcommand{\problemdef}[3]{
	\begin{center}
	\begin{minipage}{0.96\textwidth}
		\noindent
		#1
		\vspace{5pt}\\
		\setlength{\tabcolsep}{3pt}
		\begin{tabularx}{\textwidth}{@{}lX@{}}
			\textbf{Input:}     & #2 \\
			\textbf{Question:}  & #3
		\end{tabularx}
	\end{minipage}
	\end{center}
}

\DeclarePairedDelimiterX{\abs}[1]{\lvert}{\rvert}{#1}

\newcommand{\lifetime}{\ensuremath{T}}

\newcommand{\TG}{\mathcal{G}}
\newcommand{\TP}{\mathcal{P}}
\newcommand{\TGcompact}{\ensuremath{\mathcal{G}=(V, (E_i)_{i\in[\lifetime]})}}

\newcommand{\NP}{\textrm{NP}}
\newcommand{\yes}{\emph{yes}}

\newcommand{\ie}{i.\,e.,\ }

\newcommand{\TempDisjointWalks}{\textsc{Temporally Disjoint Walks}}
\newcommand{\TempDisjointPaths}{\textsc{Temporally Disjoint Paths}}
\newcommand{\STempDisjointWalks}{\textsc{Temporally Disjoint Walks}}  %
\newcommand{\STempDisjointPaths}{\textsc{Temporally Disjoint Paths}}  %

\newcommand{\NonStrVars}{\textsc{Temporally Disjoint (Paths/Walks)}}

\newcommand{\PathVersion}{\textsc{Temporally Disjoint Paths}}  %
\newcommand{\WalkVersion}{\textsc{Temporally Disjoint Walks}}  %
\newcommand{\AllProbs}{\textsc{Temporally Disjoint (Paths/Walks)}}  %

\usepackage{etoolbox}

\title{Interference-free Walks in Time: Temporally Disjoint Paths}

\author[2]{Nina Klobas\thanks{E-Mail: \texttt{nina.klobas@durham.ac.uk}}}
\author[2]{George B. Mertzios\thanks{E-Mail: \texttt{george.mertzios@durham.ac.uk}}\thanks{Supported by the EPSRC grant EP/P020372/1 and by DFG RTG~2434 while visiting TU~Berlin.}}
\author[3]{Hendrik Molter\thanks{E-Mail: \texttt{molterh@post.bgu.ac.il}}\thanks{Supported by the German Research Foundation (DFG), project MATE (NI 369/17), and by the Israeli Science Foundation (ISF), grant No.~1070/20. The main part of this work was done while the author was affiliated with TU~Berlin.}}
\author[1]{Rolf Niedermeier\thanks{E-Mail: \texttt{rolf.niedermeier@tu-berlin.de}}}
\author[1]{Philipp Zschoche\thanks{E-Mail: \texttt{zschoche@tu-berlin.de}}}

\affil[1]{Algorithmics~and~Computational~Complexity, TU Berlin, Germany}
\affil[2]{Department of Computer Science, Durham University, UK}
\affil[3]{Department of Industrial Engineering and Management, Ben-Gurion~University~of~the~Negev, Israel}

\begin{document}

\maketitle

\begin{abstract}
We investigate the computational complexity of finding \emph{temporally disjoint} 
paths or walks in temporal 
graphs. There, the edge set changes over discrete time steps
and a temporal path (resp.\ walk) uses edges that appear 
at %
monotonically increasing 
time steps.
Two paths (or walks) are temporally disjoint if they never use the
same vertex %
at the same time; otherwise, they interfere.
This reflects applications in robotics, traffic
routing, or finding safe pathways in %
dynamically changing networks.

On the one extreme, we show that on general graphs the problem is
computationally hard. The ``walk version'' is W[1]-hard when parameterized by
the number of walks.
However, it is polynomial-time solvable for any constant number of walks. The
``path version'' remains NP-hard even if we want to find
only two temporally disjoint paths.
On the other extreme, 
restricting the input temporal graph 
to have a path as underlying graph,
quite counterintuitively, we 
find NP-hardness in general 
but also identify natural
tractable cases.
\end{abstract}

\section{Introduction}
Computing (vertex-)disjoint paths in a graph is a cornerstone problem
of algorithmic graph theory. %
One of the deepest achievements in discrete mathematics,
graph minor theory~\cite{robertson1985graph,robertson1995graph}, as well as the theory 
of parameterized complexity analysis~\cite{DF13} are tightly 
connected to it.
The problem is known to be solvable in quadratic time if the number of paths is
constant, that is, it is fixed-parameter tractable when parameterized by the number of 
paths~\cite{KAWARABAYASHI2012424}.
Besides being of fundamental interest in (algorithmic) graph theory,
finding disjoint paths has many applications and there exist numerous 
variations of the problem.
In AI and robotics scenarios, for instance, multi-agent path finding 
is an intensively studied, closely related problem~\cite{Stern19,SternSFK0WLA0KB19}.

Coming from the graph-algorithmic side, we propose a new view on finding 
disjoint paths (and walks), that is, we place the problem into the 
world of temporal graphs.
We add a ``new dimension'' to the classic, static 
graph scenario by generalizing to a setting where the edges of a graph may
appear and disappear over (discrete) time. In our model, we consider two 
paths (or walks) to be disjoint if they do not use the
same vertex %
at the same point of time.
Consider \cref{fig:example} for an example.
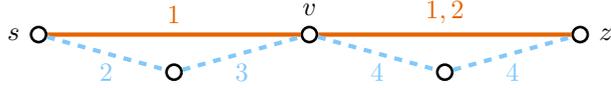
\begin{figure}
		\centering
		\begin{tikzpicture}[line width=1pt, yscale=0.5,xscale=1.8]
				\node[vert2,label=left:$s$] (s) at (0,0) {};
				\node[vert2] (1) at (1,-1) {};
				\node[vert2,label=above:$v$] (2) at (2,0) {};
				\node[vert2]  (3) at (3,-1) {};
				\node[vert2,label=right:$z$] (z) at (4,0) {};

				\path[-,draw,ultra thick,ourorange] (s) edge node[above] {$1$} (2);
				\path[-,draw,ultra thick,ourorange] (2) edge node[above] {$1,2$} (z);
				\path[-,draw,ultra thick,dashed,ourblue] (s) edge node[below] {$2$} (1);
				\path[-,draw,ultra thick,dashed,ourblue] (1) edge node[below] {$3$} (2);
				\path[-,draw,ultra thick,dashed,ourblue] (2) edge node[below] {$4$} (3);
				\path[-,draw,ultra thick,dashed,ourblue] (3) edge node[below] {$4$} (z);
		\end{tikzpicture}	
		\caption{A temporal graph where a label of an edge reflects at which time it is available. 
		There are two temporally disjoint $(s,z)$-paths $P_1$ and~$P_2$, 
		where $P_1$ uses the solid (orange) edges 
and $P_2$ the dashed (blue) edges. Here, $P_1$ uses $v$ before $P_2$.}
		\label{fig:example}
\end{figure}
Moreover, the path finding also has to take into account that edges are
not permanently available, reflecting dynamic aspects of many real-world scenarios 
such as routing in traffic or communication networks, or the very dynamic nature 
of social networks. We intend to initiate studies on this natural scenario.
Doing so, we focus on two extreme cases for the underlying graphs, namely the 
(underlying) graph structure being completely unrestricted or it being
restricted to just a path graph. For these opposite poles, 
performing (parameterized) computational complexity studies, 
we present surprising discoveries.
Before coming to these, we discuss (excerpts of) the large body of related work.

\paragraph*{Related work.}
As said, both the literature on (static) disjoint paths and its many variants
as well as on multi-agent path finding is very rich. Hence, we only list a 
small fraction of the relevant related work.
In context of graph-algorithmic work, the polynomial-time
(in\nobreakdash-)approximability of the NP-hard maximization version has been 
studied~\cite{ChuzhoyKN17}. Variants of the basic problem studied include 
bounds on the path length~\cite{Golovach2011paths} or relaxing on
the disjointness of 
paths~\cite{FluschnikKNS19,FluschnikMS19,Guo2018fast,tao2020finding}.

In directed graphs, finding two disjoint paths is already
NP-hard~\cite{fortune1980directed}, whereas in directed acyclic graphs the
problem is solvable in polynomial time for every fixed number of
paths~\cite{slivkins2010parameterized}.

As to multi-agent path planning, we remark that it has been intensively 
researched (with several possible definitions) 
in the last decade in the AI and robotics 
communities~\cite{AlmagorL20,AtzmonsSFWBZ20,Standley10,Stern19,SternSFK0WLA0KB19}.
Timing issues (concurrency of
moving agents) or various objective functions of the agents 
play a fundamental role here, and also a high variety of conflict scenarios 
is studied.\footnote{Also see %
http://mapf.info/}
The scenario we study in this work can be interpreted as a basic variant
of multi-agent path planning, now translated into the world of temporal graphs.

In algorithmic graph theory, also edge-colored graphs are studied.
Edge-colored graphs essentially are multilayer (or multiplex) graphs 
where the essential difference to temporal graphs is 
that there is no order on the graph layers.
Here, path-finding scenarios are e.g.\ motivated by applications in social 
and optical (routing) networks~\cite{Dondi2017,Santos2017multicolor,Wu2012maximum}.

Finally, as to temporal graphs, note that several prominent graph problems
have been studied in this fairly new framework.
This particularly includes research on path problems~\cite{BentertHNN20,EnrightMMZ2021,HMZ19,MertziosMS19,WuCKHHW16}. 
In particular, another model of vertex-disjoint temporal paths~\cite{KKK02},
where two temporal paths are considered vertex-disjoint if they do not visit the
same vertex. 
The problem of finding two such paths is NP-hard~\cite{KKK02}. Note
that the major difference to our model is that we allow two \emph{temporally}
disjoint paths to visit the same vertex as long as they do not both visit that
vertex at the same \emph{time}.

\paragraph*{Our contributions.}
Our results can be grouped into two parts.
First, studying temporal graphs where the underlying graph (which 
is obtained by making all temporal edges permanent) is 
unrestricted, we show that finding walks instead of paths turns out to 
be computationally easier. More specifically, finding temporally disjoint 
walks is W[1]-hard with respect to the number of walks but can be solved in
polynomial time if this number is constant 
(i.e., we develop an XP~algorithm),
whereas finding temporally disjoint paths already turns out to be NP-hard for
two paths.
Second, restricting the input to be a temporal line (i.e.,~the underlying graph to be a path), we find a polynomial-time
algorithm for a relevant special case while the problem remains NP-hard in
general (for both paths and walks). However, we also provide a fixed-parameter tractability result with respect to the number of paths.

\begin{table}[t]
\centering
\arrayrulecolor{lightgray}
\begin{tabular}{c|c|c}
\arrayrulecolor{black}
\toprule
\textsc{Temp.\ Disjoint} & \textsc{Paths} & \textsc{Walks}\\
\midrule
unrestricted & NP-hard  &  W[1]-hard wrt.\ $|S|$ \\
underlying graph & for $|S|=2$ & XP wrt.\ $|S|$ \\
\arrayrulecolor{lightgray}
\midrule
temporal line & \multicolumn{2}{l}{\hspace{1,4cm}NP-hard} \\
or tree &  FPT wrt.\ $|S|$ & \emph{open} \\
 \midrule
temporal line \& $S$  & \multicolumn{2}{c}{} \\
contains only pairs &  \multicolumn{2}{l}{\hspace{1,35cm}poly-time}  \\
of extremal points &  \multicolumn{2}{c}{}  \\
\arrayrulecolor{black}
\bottomrule
\end{tabular}
\caption{Overview computational complexity of \AllProbs. Here, $S$ is
the multiset of source-sink pairs. Temporal line means that the underlying graph is a path.}
\label{table:complexity}
\end{table}

\section{Preliminaries and problem definition}\label{sec:prelims}
We denote by~$\mathbb N$ and~$\mathbb N_0$ the natural numbers excluding and including~$0$, respectively.
An interval on $\mathbb N_0$ from $a$ to $b$ is denoted by $[a,b] \coloneqq \{ i\in \mathbb N_0 \mid a \leq i \leq b\}$
and $[a] \coloneqq [1,a]$.
\paragraph{Static graphs.}
An undirected graph~$G=(V,E)$ consists of a set~$V$ of vertices 
and a set~$E\subseteq \{\{v,w\}\mid v,w\in V,v\neq w\}$ of edges.
For a graph~$G$, we also denote by~$V(G)$ and~$E(G)$ the vertex and edge set of~$G$, respectively.
A path~$P=(V,E)$ is a graph with a set~$V(P)=\{v_1,\dots,v_k\}$ of distinct vertices and edge set~$E(P)=\{\{v_i,v_{i+1}\}\mid 1\leq i<k\}$
(we often represent path~$P$ by the tuple~$(v_1,v_2,\dots,v_k)$).
We say that~$P$ is a $(v_1,v_k)$-path and that $P$ \emph{visits} all vertices in $V(P)$.
\paragraph{Temporal graphs and temporally disjoint paths.}
A temporal graph~$\TGcompact$ consists of a set~$V$ of vertices and lifetime~$\lifetime$ many edge sets~$E_1,E_2,\dots,E_\lifetime$ over~$V$.
The pair $(e,i)$ is a \emph{time edge} of $\TG$ if $e \in E_i$.
The graph~$(V,E_i)$ is called the~$i$-th layer of $\TG$.
The \emph{underlying graph} of~$\TG$ is the static graph~$(V,\bigcup_{i=1}^\lifetime E_i)$.
A \emph{temporal ($s,z$)-walk} (or \emph{temporal walk} from $s$ to $z$) of length~$k$ from vertex $s=v_0$ to vertex $z=v_k$ in~$\TG$ is a sequence
$P = \left(\left(v_{i-1},v_i,t_i\right)\right)_{i=1}^k$
of triples %
such that for all $i\in[k]$ we have that $\{v_{i-1},v_i\}\in E_{t_i}$ and for
all $i\in [k-1]$ we have that $t_i \leq t_{i+1}$. %
The \emph{arrival time} of $P$ is $t_k$.
We say that $P$ 
\emph{visits} the vertices $V(P)\coloneqq \{ v_i \mid i\in [0,k] \}$. In particular, $P$ visits vertex $v_i$ during the time 
interval $[t_{i}, t_{i+1}]$, for all $i \in [k-1]$.
Furthermore, we say that $P$ visits $v_0$ during time interval $[t_1,t_1]$ and
$P$ visits~$v_k$ during time interval $[t_k,t_k]$.
A temporal ($s,z$)-walk $P = \left(\left(v_{i-1},v_i,t_i\right)\right)_{i=1}^k$ of length $k$ 
is called a \emph{temporal ($s,z$)-path} (or \emph{temporal path} from $s$ to $z$) if~$v_i\neq v_j$ whenever $i\neq j$.
Given two temporal walks $P_1, P_2$ 
we say that $P_1$ and $P_2$ \emph{temporally intersect} 
if there exists a vertex $v$ and two time intervals $[a_1,b_1],[a_2,b_2]$, where $[a_1,b_1]\cap[a_2,b_2]\neq \emptyset$, 
such that $v$ is visited by $P_1$ during~$[a_1,b_1]$ and by~$P_2$ during~$[a_2,b_2]$. 
Now, we can formally define our problem.
\problemdef{\TempDisjointPaths}{A temporal graph~$\TGcompact$ 
and a multiset~$S$ of
source-sink pairs containing elements from $V\times V$.}
{Are there pairwise
temporally non-intersecting %
temporal $(s_i,z_i)$-paths for all~$(s_i,z_i)\in S$?}
Analogously, \TempDisjointWalks{} gets the same input but asks whether
there are pairwise
temporally non-intersecting %
temporal $(s_i,z_i)$-walks for all~$(s_i,z_i)\in S$.
From the NP-hardness of \textsc{Disjoint Paths}~\cite{karp1975computational}, we
immediately get the following.
\begin{observation}
\NonStrVars\ is NP-hard even if $T=1$.
\end{observation}

 For an instance $(\TGcompact, S)$ of \AllProbs{} 
 we assume throughout this paper that in the input, 
 $\TG$ is given by the vertex set $V$ followed by the ordered (by time label) 
 subsequence of $E_1,\dots,E_\lifetime$, only containing the non-empty edge sets. 
 To make the presentation simpler, we apply a linear-time preprocessing step to the input 
 by renumbering these non-empty edge sets $E_i$ (still keeping their relative order the same) 
 such that all non-empty sets are consecutive. 
 Note that this preprocessing step creates an instance that is equivalent to the original input instance. 
 Therefore, without loss of generality we assume in the remainder of the paper that all edge sets $E_i$, $i=1,2,\ldots,T$, are non-empty 
 (where the new lifetime $T$ is now the number of non-empty edge sets in the original input).
Hence, the \emph{size} of $\TG$ is $|\TG| \coloneqq |V| + \sum_{t=1}^\lifetime |E_t|$. 

\paragraph{Parameterized complexity.} 
Let~$\Sigma$ denote a
finite alphabet.
A parameterized problem~$L\subseteq \{(x,k)\in \Sigma^*\times \mathbb N_0\}$ is a subset of all instances~$(x,k)$ from~$\Sigma^*\times \mathbb N_0$,
where~$k$ denotes the \emph{parameter}.
A parameterized problem~$L$ is 
\begin{inparaenum}[(i)]
\item FPT (\emph{fixed-parameter tractable}) if there is an algorithm that decides every instance~$(x,k)$ for~$L$ in~$f(k)\cdot |x|^{O(1)}$ time, and 
 \item contained in the class XP if there is an algorithm that decides every instance~$(x,k)$ for~$L$ in~$|x|^{f(k)}$ time.  %
\end{inparaenum}
where~$f$ is any computable function only depending on the parameter.
If a parameterized problem $L$ is W[1]-hard, then it is presumably not
fixed-parameter tractable. 
We refer to \citet{DF13} for more details.

\section{The case of few source-sink pairs}\label{sec:fewpairs}
In this section, we study the computational 
complexity of \AllProbs\ for the case
that the size of the multiset~$S$ of 
source-sink pairs is small. We start by showing
that \TempDisjointPaths{} is NP-hard even for two sink-source pairs. This is a
similar situation as for finding vertex-disjoint paths in directed static graphs,
which is also NP-hard for two paths~\cite{fortune1980directed}. However,
in the temporal setting there is a surprising difference between finding walks and
paths that does not have an analogue in the static setting. We show that
\WalkVersion\ is W[1]-hard for the number $|S|$ of source-sink
pairs and is contained in XP for the same parameter.

\begin{restatable}{theorem}{thmnph}\label{thm:NP-h-for-3-path}
		\TempDisjointPaths{} is NP-hard 
		even if $|S|=2$ and $\lifetime=3$.
\end{restatable}
\begin{proof}%
		We show that \TempDisjointPaths{} is NP-hard 
		even if $|S|=2$ and $\lifetime=3$ by a polynomial-time 
		reduction from the NP-complete 
		\textsc{Exact $(3,4)$-SAT} problem~\cite{Tov84}. 
		\textsc{Exact $(3,4)$-SAT} asks whether 
		a Boolean formula~$\phi$ is satisfiable, 
		assuming that it is in conjunctive normal form, 
		each clause has exactly three literals, 
		and each variable appears in exactly four clauses.
\begin{figure}
		\centering
	\begin{tikzpicture}[line width=1pt, scale=.32, yscale=1]
			\begin{scope}
			\draw[thin,dashed] (-4,-2) -- (21,-2);
					\node at (-4,0) {$G_1$:};
			\node[vert2,label=below:$s_2$] (s2) at (-1,0) {};
			\node[vert2,label=$s_1$] (s1) at (0,0) {};
			\begin{scope}[xshift=1cm]
			\node[vert2] (a1) at (0,0) {};
			\node[vert2,] (x1) at (1,1) {};
			\node[vert2,] (x2) at (2,1) {};
			\node[vert2,] (x3) at (3,1) {};
			\node[vert2,] (x4) at (4,1) {};

			\node[vert2,] (nx1) at (1,-1) {};
			\node[vert2,] (nx2) at (2,-1) {};
			\node[vert2,] (nx3) at (3,-1) {};
			\node[vert2,] (nx4) at (4,-1) {};
		\end{scope}

		\begin{scope}[xshift=6cm]
				\node[vert2] (a2) at (0,0) {};
			\node[vert2,] (y1) at (1,1) {};
			\node[vert2,] (y2) at (2,1) {};
			\node[vert2,] (y3) at (3,1) {};
			\node[vert2,] (y4) at (4,1) {};

			\node[vert2,] (ny1) at (1,-1) {};
			\node[vert2,] (ny2) at (2,-1) {};
			\node[vert2,] (ny3) at (3,-1) {};
			\node[vert2,] (ny4) at (4,-1) {};
		\end{scope}
		\begin{scope}[xshift=11cm]
			\node[vert2] (a3) at (0,0) {};
			\node[vert2,] (z1) at (1,1) {};
			\node[vert2,] (z2) at (2,1) {};
			\node[vert2,] (z3) at (3,1) {};
			\node[vert2,] (z4) at (4,1) {};

			\node[vert2,] (nz1) at (1,-1) {};
			\node[vert2,] (nz2) at (2,-1) {};
			\node[vert2,] (nz3) at (3,-1) {};
			\node[vert2,] (nz4) at (4,-1) {};
		\end{scope}
			\node[vert2] (a4) at (16,0) {};
		\draw (s1) -- (a1);
		\draw (a1) -- (x1) -- (x2) -- (x3) -- (x4) -- (a2) -- (nx4) -- (nx3) -- (nx2) -- (nx1) -- (a1) ;
		\draw (a2) -- (y1) -- (y2) -- (y3) -- (y4) -- (a3) -- (ny4) -- (ny3) -- (ny2) -- (ny1) -- (a2) ;
		\draw (a3) -- (z1) -- (z2) -- (z3) -- (z4) -- (a4) -- (nz4) -- (nz3) -- (nz2) -- (nz1) -- (a3) ;

		\draw[dashed] (a4) -- (17,1);
		\draw[dashed] (a4) -- (17,-1);

		\node[vert2] (a5) at (19,0) {};

		\draw[dashed] (a5) -- (18,1);
		\draw[dashed] (a5) -- (18,-1);

		\node[vert2,label=right:$z_2$] (t_2) at (20,0) {};
		\node[vert2,label=right:$s'$] (sx) at (20,-1) {};
		\draw (a5) -- (sx);
		\end{scope}

		\begin{scope}[yshift=-4cm]
			\draw[thin,dashed] (-4,-2) -- (21,-2);
					\node at (-4,0) {$G_2$:};
			\node[vert2,label=below:$s_2$] (s2) at (-1,0) {};
			\node[vert2,label=$s_1$] (s1) at (0,0) {};
			\begin{scope}[xshift=1cm]
			\node[vert2] (a1) at (0,0) {};
			\node[vert2,] (x1) at (1,1) {};
			\node[vert2,] (x2) at (2,1) {};
			\node[vert2,] (x3) at (3,1) {};
			\node[vert2,] (x4) at (4,1) {};

			\node[vert2,] (nx1) at (1,-1) {};
			\node[vert2,] (nx2) at (2,-1) {};
			\node[vert2,] (nx3) at (3,-1) {};
			\node[vert2,] (nx4) at (4,-1) {};
		\end{scope}

		\begin{scope}[xshift=6cm]
			\node[vert2] (a2) at (0,0) {};
			\node[vert2,] (y1) at (1,1) {};
			\node[vert2,] (y2) at (2,1) {};
			\node[vert2,] (y3) at (3,1) {};
			\node[vert2,] (y4) at (4,1) {};

			\node[vert2,] (ny1) at (1,-1) {};
			\node[vert2,] (ny2) at (2,-1) {};
			\node[vert2,] (ny3) at (3,-1) {};
			\node[vert2,] (ny4) at (4,-1) {};
		\end{scope}
		\begin{scope}[xshift=11cm]
			\node[vert2] (a3) at (0,0) {};
			\node[vert2,] (z1) at (1,1) {};
			\node[vert2,] (z2) at (2,1) {};
			\node[vert2,] (z3) at (3,1) {};
			\node[vert2,] (z4) at (4,1) {};

			\node[vert2,] (nz1) at (1,-1) {};
			\node[vert2,] (nz2) at (2,-1) {};
			\node[vert2,] (nz3) at (3,-1) {};
			\node[vert2,] (nz4) at (4,-1) {};
		\end{scope}
			\node[vert2] (a4) at (16,0) {};
			\draw (s2) -- (nx1) -- (nx2) -- (nx3) -- (nx4) -- 
			(ny1) -- (ny2) -- (ny3) -- (ny4) --
			(nz1) -- (nz2) -- (nz3) -- (nz4);
			\draw (a1) edge[bend right=62] (18,-1);
			\draw[dashed] (17,-1) -- (18,-1);
			\draw[dashed] (nz4) -- (16,-1);
		\draw (a1) -- (x1) -- (x2) -- (x3) -- (x4) -- (a2) ;
		\draw (a2) -- (y1) -- (y2) -- (y3) -- (y4) -- (a3)  ;
		\draw (a3) -- (z1) -- (z2) -- (z3) -- (z4) -- (a4)  ;

		\draw[dashed] (a4) -- (17,1);

		\node[vert2] (a5) at (19,0) {};

		\draw[dashed] (a5) -- (18,1);

		\node[vert2,label=right:$z_2$] (t2) at (20,0) {};
		\node[vert2,label=right:$s'$] (sx) at (20,-1) {};
		\draw (a5) -- (t2);
		\end{scope}

		\begin{scope}[yshift=-11.5cm]
					\node at (-4,0) {$G_3$:};
			\node[vert2,label=below:$s_2$] (s2) at (-1,0) {};
			\node[vert2,label=$s_1$] (s1) at (0,0) {};
			\begin{scope}[xshift=1cm]
			\node[vert2] (a1) at (0,0) {};
			\node[vert2,] (x1) at (1,1) {};
			\node[vert2,] (x2) at (2,1) {};
			\node[vert2,] (x3) at (3,1) {};
			\node[vert2,] (x4) at (4,1) {};

			\node[vert2,] (nx1) at (1,-1) {};
			\node[vert2,] (nx2) at (2,-1) {};
			\node[vert2,] (nx3) at (3,-1) {};
			\node[vert2,label={[label distance=-4pt]above:$x_1^4$}] (nx4) at (4,-1) {};

		\end{scope}

		\begin{scope}[xshift=6cm]
			\node[vert2] (a2) at (0,0) {};
			\node[vert2,] (y1) at (1,1) {};
			\node[vert2,] (y2) at (2,1) {};
			\node[vert2,] (y3) at (3,1) {};
			\node[vert2,] (y4) at (4,1) {};

			\node[vert2,] (ny1) at (1,-1) {};
			\node[vert2,] (ny2) at (2,-1) {};
			\node[vert2,label={[label distance=-4pt]above:$x_2^3$}] (ny3) at (3,-1) {};
			\node[vert2,] (ny4) at (4,-1) {};
		\end{scope}
		\begin{scope}[xshift=11cm]
			\node[vert2] (a3) at (0,0) {};
			\node[vert2,] (z1) at (1,1) {};
			\node[vert2,] (z2) at (2,1) {};
			\node[vert2,label={[label distance=-2pt]above:$\overline{x_2}^3$}] (z3) at (3,1) {};
			\node[vert2,] (z4) at (4,1) {};

			\node[vert2,] (nz1) at (1,-1) {};
			\node[vert2,] (nz2) at (2,-1) {};
			\node[vert2,] (nz3) at (3,-1) {};
			\node[vert2,] (nz4) at (4,-1) {};
		\end{scope}
			\node[vert2] (a4) at (16,0) {};

		\node[vert2] (a5) at (19,0) {};

		\node[vert2,label=right:$z_2$] (t2) at (20,0) {};
		\node[vert2,label=right:$s'$] (sx) at (20,-1) {};

		\begin{scope}[yshift=-1cm]
		\node[vert2,label=below:$C_1$] (C1) at (18,-3) {};
		\node[vert2,label=below:$C_{m+1}$] (Cm) at (0,-3) {};
		\draw (C1) -- (sx);
		\draw[dashed] (C1) -- (17,-1.5);
		\draw[dashed] (C1) -- (16,-1.5);
		\draw[dashed] (C1) -- (15,-2);

		\draw[dashed] (Cm) -- (1,-1.5);
		\draw[dashed] (Cm) -- (2,-1.5);
		\draw[dashed] (Cm) -- (3,-2);

		\node[vert2,label=below:$C_i$] (Ci) at (11,-3) {};
		\draw[dashed] (Ci) -- (12,-1.5);
		\draw[dashed] (Ci) -- (13,-1.5);
		\draw[dashed] (Ci) -- (14,-2);

		\node[vert2,label=below:$C_{i+1}$] (Ci1) at (7,-3) {};
		\draw[dashed] (Ci1) -- (6,-1.5);
		\draw[dashed] (Ci1) -- (5,-1.5);
		\draw[dashed] (Ci1) -- (4,-2);

		\draw (Ci) -- (z3) -- (Ci1);
		\draw (Ci) -- (ny3) -- (Ci1);
		\draw (Ci) -- (nx4) -- (Ci1);

		\node[vert2,label=left:$z_1$] (t1) at (-2,-3) {};
		\draw (Cm) -- (t1);
		\end{scope}
		\end{scope}
\end{tikzpicture}

\caption{An excerpt of layers $G_j =(V,E_j), j \in [3]$ of the temporal
graph $\TG$ used in \cref{thm:NP-h-for-3-path}. Clause $C_i = (x_1
\vee x_2 \vee \overline{x_3})$ contains the $4$th appearance of $x_1$ and the $3$rd appearance of $x_2$ and~$x_3$.}
		\label{fig:np-h-k-2}
		\end{figure}
		\paragraph{Construction.}
		Let $\phi$ be an instance of 
		\textsc{Exact $(3,4)$-SAT} with $n$ variables $x_1,x_2,\dots, x_n$ and $m$ clauses.
		We construct an 
		instance~$I = (\TG =(V,(E_1,E_2,E_3)), S=\{ (s_1,z_1),(s_2,z_2) \})$ in the
		following way.
		Intuitively, the first two layers contain the assignment gadget for the
		variables.
		The part of the temporal $(s_1,z_1)$-path $P_1$ which is in the first layer sets all variables.
		The temporal $(s_2,z_2)$-path ensures that $P_1$ does not ``wait'' in a
		variable gadget.
		In the third layer, $P_1$~must go from $s'$ to $z_1$ through all clause
		gadgets.
		Since $P_1$ cannot visit a vertex twice, this validates whether the assignment 
		satisfies $\phi$. \cref{fig:np-h-k-2} depicts the resulting temporal graph.

		The construction is done as follows. For each variable~$x_j$ we construct the 
		\emph{variable gadget} $\TG^{x_j}=(\{a_{x_j},a_{x_{j+1}}\}
		\cup \{{x_j}^i,\overline{{x_j}}^i\mid i \in [4]\},(E^{x_j}_i)_{i\in[2]})$,
		where $E^{x_j}_1 = E^{x_j}_T \cup E^{x_j}_F$ with 
		$E^{x_j}_T = \{ \{a_{x_j},{x_j}^1\},$ $\{{x_j}^4,a_{x_{j+1}}\} \} \cup \{\{{x_j}^i,{x_j}^{i+1}\} \mid i\in[3]\}$ and 
$E^{x_j}_F = \{ \{a_{x_j},{\overline{x_j}}^1\},\{\overline{x_j}^4,a_{x_{j+1}}\} \} \cup \{\overline{x_j}^i,\overline{x_j}^{i+1} \mid i\in[3]\}$,
and $E^{x_j}_2 = E^{x_j}_T \cup \{\{\overline{x_j}^i,\overline{x_j}^{i+1}\} \mid
		i\in[3]\} \cup \{\{\overline{x_j}^4,\overline{x_{j+1}}^{1}\} \mid
j<n\}$. 
		Let $C_i = (\ell_j^f \vee \ell_p^g \vee \ell_q^h)$ be a clause,
		where $\ell_\alpha^\beta$ is a literal of variable $x_\alpha$ and its $\beta$-th appearance
		when iterating the clauses in the order of the indices.
		We now abuse our notation and say $\ell_\alpha^\beta \equiv \overline{x}_\alpha^\beta$ if $\ell_\alpha^\beta$ is a negation of $x_\alpha$ and otherwise we say $\ell_\alpha^\beta \equiv x_\alpha^\beta$.
		We construct the 
		\emph{clause gadget} $\TG^{C_i} := (\{C_i,C_{i+1},\ell_j^f, \ell_p^g, \ell_q^h\},(E^{C_i}_t)_{t\in[3]})$ 
		where $E^{C_i}_1=E^{C_i}_2=\emptyset$ and $E^{C_i}_3 = \{ \{C_r,\ell_j^f\},\{C_r,\ell_p^g\},\{C_r,\ell_q^h\} \mid r \in \{i,i+1\} \}$.
		Now we set $\TG =(V,(E_t)_{t\in[3]})$, where 
		$V=\bigcup_{j=1}^n V(\TG^{x_j})\cup \bigcup_{i=1}^m V(\TG^{C_i}) \cup \{ s_1,s_2, z_1,z_2, s'\}$,
		$E_1 = \bigcup_{j=1}^n E^{x_j}_1 \cup \{\{s_1,a_{x_1}\},\{a_{x_{n+1}},s'\} \}$,
		$E_2 = \bigcup_{j=1}^n E^{x_j}_2 \cup
		\{\{s_2,{x_1}^1\},\{{x_n}^4,a_{x_1}\},\{a_{x_{n+1}},z_2\} \}$, and %
		$E_3 = \bigcup_{i=1}^m E^{C_i}_4 \cup \{ \{s',C_1\},\{z_1,C_{m+1}\} \}$.

		Observe that $I$ can be constructed in polynomial time. 
		\paragraph{Correctness.}
		($\Leftarrow$): Assume $I$ is a \yes-instance.
		Hence, there are non-intersecting $P_1,P_2$ such that $P_i$ is a temporal $(s_i,z_i)$-path, for $i \in[2]$.
		Observe that there is only one unique temporal $(s_2,z_2)$-path in $\TG$. %
		Hence, $P_1$ visits~$s'$, otherwise $P_1$ intersects~$P_2$.
		Moreover, on its way to~$s'$ the temporal path~$P_1$ visits for each variable~$x_i$ 
		either $x_i^1,x_i^2,x_i^3,x_i^4$
		or $\overline{x}_i^1,\overline{x}_i^2,\overline{x}_i^3,\overline{x}_i^4$.
		We set variable $x_i$ to true if and only if $P_1$ visits $\overline{x}_i^1$.
		Now observe that $P_1$ visits 
		all vertices $C_1,C_2,\dots C_{m+1}$ (in that order) at time $3$.
		In particular, the only way from $C_i$ to $C_{i+1}$ via one of literals of $C_i$.
		Since $P_1$ cannot visit a vertex twice, one of the literals is set to true by our assignment of the variables.
		Thus, $\phi$ is satisfiable.

		($\Rightarrow$): Assume $\phi$ is satisfiable.
		We fix one assignment for the variables which satisfy $\phi$.
		Let $P_2$ be the unique temporal $(s_2,z_2)$-path in $\TG$. %
		Our $P_1$ starts with the time edge $(\{s_1,a_{x_1}\},1)$ and 
		to get from $a_{x_1}$ to $a_{x_{n+1}}$ at time one, 
		$P_1$ takes for each $j \in [n]$ the 
		induced path $(V,E_1)[\{a_{x_j},a_{x_{j+1}}\} \cup \{x_j^i \mid i\in[4]\}]$ 
		if $x_j$ is set to false, otherwise $P_1$ takes the 
		induced path $(V,E_1)[\{a_{x_j},a_{x_{j+1}}\} \cup \{\overline{x_j}^i \mid i\in[4]\}]$.
		Afterwards, $P_1$ takes the time edges $(\{a_{x_{n+1}},s'\},1),(\{s',C_1\},1)$.
		Let $i \in [m]$ and $C_i = (\ell_j^f \vee \ell_p^g \vee \ell_q^h)$ be a clause,
		where $\ell_\alpha^\beta$ is a literal of variable $x_\alpha$ and its
		$\beta$-th appearance. To get from $C_i$ to $C_{i+1}$ at time three,
		$P_1$ takes the 
		induced path $(V,E_3)[\{C_i,C_{i+1}, \ell_\alpha^\beta\}]$, where $\ell_\alpha^\beta$ is a satisfied literal in $C_i$.
		Finally, $P_1$ takes the time edge $(\{C_{m+1},z_1\},3)$.
		Note that $P_1$ does not visit a vertex twice, and hence is a temporal $(s_1,z_1)$-path in $\TG$.
		Moreover, it visits $s'$ at times $\{1,2,3\}$ and all other vertices either
		at time one or time three. Note that $P_2$ is a temporal path that
		exclusively visits vertices at time two and does not visit vertex $s'$.
		Thus, $P_1$ is temporally non-intersecting with $P_2$. It follows that $I$ is a
		\yes-instance.
\end{proof}

The reduction behind \cref{thm:NP-h-for-3-path} heavily relies on the fact that
we are dealing with paths. Indeed, for temporally disjoint walks we presumably
cannot obtain NP-hardness for a constant number of sink-source pairs since, as
we will show at the end of this section, \WalkVersion\ can be solved in
polynomial time if the number of source-sink pairs is constant. However, 
before that we show
W[1]-hardness for \WalkVersion\ parameterized by the number $|S|$ of source-sink
pairs, presumably excluding the existence of an FPT-algorithm for this
parameter.

\begin{restatable}{theorem}{thmwh}%
		\label{thm:w1}
\WalkVersion\ is W[1]-hard when parameterized by $|S|$, even if all edges have exactly one time label.
\end{restatable}
\begin{proof}
We present a parameterized reduction from \textsc{Disjoint Paths on DAGs}. In
this problem, we are given a DAG $D=(U,A)$ and $\ell$ source-sink pairs
$\{(s_1, t_1), \ldots, (s_\ell, t_\ell)\}\subseteq U\times U$, and we are asked
whether there exist $\ell$ paths $P_i$ from $s_i$ to $t_i$ that are pairwise
vertex-disjoint. This problem parameterized by the number $\ell$ of paths is
W[1]-hard~\cite{slivkins2010parameterized}\footnote{Slivkins~\cite{slivkins2010parameterized}
shows W[1]-hardness for the arc-disjoint version of this problem. However,
there are straightforward (parameterized) reductions between \textsc{Disjoint
Paths on DAGs} and \textsc{Arc-Disjoint Paths on DAGs}.}.

\paragraph{Construction.} Let $(D=(U,A), \{(s_1, t_1), \ldots,
(s_\ell, t_\ell)\}\subseteq U\times U)$ be an instance of \textsc{Disjoint Paths on DAGs}. We first compute a topological
ordering $\sigma$ for the vertices in $U$~\cite[Theorem~4.2.1]{BJG02}. Recall that $\sigma$ is a linear ordering on
the vertices in $U$ with the property that $(u,v)\in A\Rightarrow u <_\sigma
v$. We denote with $\sigma(u)$ the position of vertex $u$ in the ordering~$\sigma$.

Using this ordering, we construct a temporal graph $\TG$ as follows. We add the
following vertices to $V$:
\begin{itemize}
\item For every vertex $u\in U$ we add a ``vertex-vertex'' $x_u$ to~$V$.
\item For every arc $(u,v)\in A$ we add an ``arc-vertex'' $y_{(u,v)}$ to~$V$.
\end{itemize}
Next, we add the following time edges: For every arc $(u,v)\in A$ we add an
edge from $x_u$ to $y_{(u,v)}$ with time label $2\sigma(u)-1$ and we add an
edge from $y_{(u,v)}$ to $x_v$ with label $2(\sigma(v)-1)$.

Finally, for every source-sink pairs $(s_i, t_i)$ of the \textsc{Disjoint Paths
on DAGs} instance, we add the source-sink pair $(x_{s_i}, x_{t_i})$ to $S$. The
reduction can clearly be performed in polynomial time and the number of
source-sink pairs in the constructed instance is the same as the number of
source-sink pairs in the given \textsc{Disjoint Paths on DAGs} instance.

\paragraph{Correctness.} We now show that the given \textsc{Disjoint Paths on
DAGs} instance is a \yes-instance if and only if our constructed instance is a
\yes-instance of \WalkVersion.

$(\Rightarrow)$: Assume $(D=(U,A), \{(s_1, t_1), \ldots, (s_\ell, t_\ell)\}\subseteq
U\times U)$ is a \yes-instance of \textsc{Disjoint Paths on DAGs}. Then there is
a set $\mathcal{P}$ of pairwise vertex-disjoint paths for all source-sink pairs
$\{(s_1, t_1), \ldots, (s_\ell, t_\ell)\}$. Let $P_i=((v_{j-1},v_j))_{j=1}^k\in
\mathcal{P}$ with $v_0=s_i$ and $v_k=t_i$ be the path from $s_i$ to $t_i$. Then,
by construction of $\TG$, we have that 
$Q_i=((x_{s_i}, y_{(s_i,v_1)}, 2\sigma(s_i)-1), (y_{(s_i,v_1)}, x_{v_1}, 2(\sigma(v_1)-1)),$ 
$(x_{v_1}, y_{(v_1,v_2)}, 2\sigma(v_1)-1), \ldots, 
(y_{(v_{k-1},t_i)}, x_{t_i}, 2(\sigma(t_i)-1)))$ is a strict temporal path from $x_{s_i}$ to $x_{t_i}$ in
$\TG$ that alternatingly visits the vertex-vertices and arc-vertices that correspond to the
vertices and arcs visited by $P_i$ in $D$. Since the paths $P_i\in\mathcal{P}$
are pairwise vertex disjoint (and hence also arc-disjoint), it is clear that the
temporal paths $\{Q_i\}_{i=1}^\ell$ are vertex-disjoint, and hence also
temporally disjoint.

$(\Leftarrow)$: 
Note that a 
temporal walk can visit vertices multiple times. But there are only
two ways in which a temporal walk $W$ visits a vertex $x_{v_j}$
multiple times: either $W$ contains subwalk 
$((x_{v_j}, y_{(v_j,v_{j'})}, 2\sigma(v_j)-1),(y_{(v_j,v_{j'})},x_{v_j}, 2(\sigma(v_j)-1)))$ for
some $j'$, or $W$ contains the subwalk 
$((x_{v_j}, y_{(v_j,v_{j'})},2\sigma(v_{j'})-1),$ $(y_{(v_j,v_{j'})},x_{v_j},2(\sigma(v_{j'})-1)))$
for some $j'$. In both cases we can remove the subwalk from $W$ and still obtain
a temporal walk from the same starting vertex to the same end vertex
than $W$. We can remove subwalks of this kind repeatedly until we obtain a
temporal path. This observations above allow us to assume that if we face a
\yes-instance of \WalkVersion, then there is a set $\mathcal{P}$ of pairwise
temporally disjoint \emph{paths} for all source-sink pairs in~$S$.

Now observe
that all vertices in $\TG$ are $x_{v_j}$ incident with time edges of exactly two
time labels: $2\sigma(v_j)-1$ and $2(\sigma(v_{j})-1)$. Hence, our paths ``enter'' a vertex-vertex $x_{v_j}$ with a time edge that has time
label $2\sigma(v_j)-1$ and ``leave''~$x_{v_j}$ with a time edge that has time
label $2(\sigma(v_{j})-1)$. It follows that no two of the temporally
disjoint paths in $\mathcal{P}$ visit the same vertex-vertex $x_{v_j}$. The same
holds for arc-vertices~$y_{(v_j,v_{j'})}$ since they are only incident with two
time edges. It follows that the temporal paths in $\mathcal{P}$ are pairwise
vertex-disjoint. Furthermore, by the construction of $\TG$, they alternatingly
visit vertex-vertices and arc-vertices, where the latter correspond to arcs 
in~$D$ that connect the vertices corresponding to the vertex-vertices visited
directly before and after. It follows that we can translate temporal paths in
$\mathcal{P}$ directly to a set of pairwise vertex-disjoint paths in $D$ that
connect the corresponding source-sink pairs of the \textsc{Disjoint Paths on
DAGs} instance.
\end{proof}

We now show that \WalkVersion{} is in XP for the parameter number
$|S|$ of source-sink pairs.

\begin{restatable}{theorem}{thmxp}%
		\label{thm:xp}
		\WalkVersion{} is in the class XP when parameterized by $|S|$, 
		as it can be solved in $O(|V|^{2|S|+2}\cdot\lifetime)$ time if $|S|$ is a fixed constant.
\end{restatable}
\begin{proof}%
		Consider an instance
		$I=(\TG,S=\{(s_1,z_1),(s_2,z_2),\dots,(s_k,z_k)\})$ of \TempDisjointWalks{}.
		We use the following dynamic programming table~$D$ with Boolean entries.
		Intuitively, we want that for all $t \in \{1,\dots,T\}$ and $v_1,\dots,v_k
		\in V$ we have that $D[t,v_1,\dots,v_k] = \top$
		if and only if
		there are temporally non-intersecting temporal $(s_i,v_i)$-walks
		$P_1,\dots,P_k$ with arrival time $t_i\leq t$. However, for technical
		reasons, we have slightly stronger requirements for $D$. First of all, we have
		a ``dummy'' time label zero that we use to encode the sources in the dynamic
		program. Formally, we initialize $D$ as follows:
		
		For all $v_1,\dots v_k \in V$ we have that
		\begin{align*}
				D[0,v_1,\dots,v_k] \coloneqq
				\begin{cases}
						\top, & \text{if }\forall i\in[k]\colon v_i=s_i\\
						\bot, & \text{otherwise.}
				\end{cases}
		\end{align*}

		Furthermore, we have to model that the temporal walks we are looking for do
		not have to start immediately at their respective sources. Hence, if a
		temporal walk is still ``waiting'' at its source, the source vertex is not
		``blocked'' for other temporal walks. We have a symmetric situation if temporal
		walks already arrived at their respective sink. In other words, if we have an
		entry $D[t,v_1,\dots,v_k]$ with $v_i=v_j$ for some $i\neq j$, then it is a
		necessary condition for $D[t,v_1,\dots,v_k]=\top$ that at least one of the two
		temporal walks $i,j$ is either still waiting at its source or already arrived
		at its sink. In the latter case, we additionally need that the temporal walk
		arrived at the sink in a previous time step, otherwise the sink would still be
		blocked. We now look up in $D$ whether there all these conditions are met for
		a set of temporally disjoint walks that arrive at some vertices at time $t-1$ such that they can be
		extended in time step $t$ to reach the vertices $v_1,\dots,v_k$.
		
		Formally, for all $t\in
		[\lifetime]$ we have that
		\begin{align*}
				D[t,v_1,\dots,v_k] \coloneqq 
				\begin{cases}
						\top, & \begin{array}{l}
								\text{if }\forall i,j\in[k],v_i=v_j\colon\\
								\ \ \exists p\in\{i,j\}\colon v_p \in \{s_p,z_p\},\\
									\exists u_1,\dots,u_k \in V \colon\\
									\ \ (\forall i,j\in[k],i\not=j,v_i=z_i\colon \\
									\ \ \ \  u_i = z_i \vee v_i\not=v_j) \text{ and}\\
									\ \ D[t-1,u_1,\dots,u_k]=\top \text{ and}\\
									\ \ ((V,E_t), \{ (u_i,v_i) \mid i\in[k],\\
									\ \ \{u_i,v_i\}\not=\{z_i\},v_i \not=s_i\})
									\ \ \text{is a}\\
									\ \ \text{\yes-instance of \textsc{Disj.~Paths},}
								\end{array}\\ 
						\bot, &\text{otherwise,}
				\end{cases}
		\end{align*}
		Here, \textsc{Disj.~Paths} is the \textsc{Disjoint Paths} problem, 
		where we are given an undirected graph $G$ and set of $k$ terminal pairs $\{(s_i',z_i') \mid i \in [k]\}$
		and ask whether there are $k$ vertex-disjoint paths $P_1,\dots,P_k$ in $G$ 
		such that $P_i$ is an $(s_i',z_i')$-path 
		and $P_i$ and $P_j$ are vertex-disjoint for all $i \in [k]$ and $j\in[k] \setminus \{i\}$.
		Note that $s_i'=z_i'$ is a valid input and 
		that in this case $s_i'$ is the only vertex on an $(s_i',z_i')$-path.
		We report that $I$ is a \yes-instance if and only if $D[T,z_1,\dots,z_k] = \top$.
		\paragraph{Correctness.}We show by induction that for all $t \in \{0,1,\dots,T\}$ and $v_1,\dots,v_k \in V$ we have that 
		$D[t,v_1,\dots,v_k] = \top$
		if and only if
		there are temporally non-intersecting temporal walks $P_1,\dots,P_k$ 
		such that 
		(a) $P_i$ is a temporal $(s_i,v_i)$-walk in $\TG$ with arrival time $t_i\leq t$, and
		(b) $\forall i,j \in [k],v_i\not=v_j\colon \exists p\in\{i,j\}\colon v_p=s_p \vee (v_p=z_p \wedge t_p<t)$ 
		is true.

		Here, we say that an empty edge list $(\emptyset)$ is a temporal $(v,v)$-walk with 
		arrival time $0$ which does not visit any vertex, for all $v \in V$.
		Hence, for $t=0$ this claim is true by the definition of $D$.
		Now let $t \in [T]$ and assume that this claim holds true for all $t'<t$.

		($\Leftarrow$): Assume that there are temporally 
		non-intersecting temporal walks $P_1,\dots,P_k$ 
		such that (a) and (b) hold true.
		Let $K_s:=\{ i \in [k] \mid s_i=v_i\}$, $K_z:=\{ i \in [k] \mid z_i=v_i, t_i < t-1 \}$, and $K:=[k]\setminus(K_s \cup K_z)$.	
		We set $u_i = s_i$, for all $i\in K_s$.
		We set $u_i = z_i$, for all $i\in K_z$.
		For all $i \in K$, 
		let $u_i$ be the first vertex which is visited at time~$t$ (in the visiting order),
		or $v_i$ if $P_i$ does not visit any vertex at time $t$.
		Observe that $\forall i,j\in[k],v_i=v_j\colon\exists p\in\{i,j\}\colon v_p \in \{s_p,z_p\}$ 
		and $\forall i,j\in[k],i\not=j,v_i=z_i\colon u_i = z_i \vee v_i\not=v_j$ are both true.
		Observe for all $i \in K$ that we can split $P_i$ into two parts $P_i^1$ and $P_i^2$, where
		$P_i^1$ is a (possibly empty) temporal $(s_i,u_i)$-walk in $\TG$ with arrival time at most $t-1$ and
		$P_i^2$ is %
		a (possibly empty) $(u_i,v_i)$-path in the graph $(V,E_t)$. 
		Note that $\{ P_i^2 \mid i\in K\}$ is a solution
		for the \textsc{Disjoint Path} instance $((V,E_t),\{ (u_i,v_i) \mid i\in[k],\{u_i,v_u\}\not=\{z_i\},v_i \not=s_i\})$.
		We set $P_i^1= (\emptyset)$ for all $i \in K_s$, and $P_i^1 = P_i$ for all $i \in K_z$.
		Clearly, $P_1^1,\dots,P_k^1$ are temporally non-intersecting such that
		(a) $P_i$ is a temporal $(s_i,u_i)$-walk in $\TG$ with arrival time $t_i\leq t-1$, and
		(b) $\forall i,j \in [k],u_i\not=u_j\colon \exists p\in\{i,j\}\colon u_p=s_p \vee (u_p=z_p \wedge t_p<t-1)$.
		Hence, $D[t-1,u_1,\dots,u_k] = \top$ and thus by definition of $D$ we have $D[t,v_1,\dots,v_k] = \top$.

		($\Rightarrow$): Assume that $D[t,v_1,\dots,v_k] = \top$.
		By definition of $D$, we know that there are vertices $u_1,\dots, u_k$ such that 
		(i) $\forall i,j\in[k],v_i=v_j\colon\exists p\in\{i,j\}\colon v_p \in \{s_p,z_p\}$ 
		(ii) $\forall i,j\in[k],i\not=j,v_i=z_i\colon u_i = z_i \vee v_i\not=v_j$ are true,
		(iii)  $D[t-1,u_1,\dots,u_k] = \top$, and
		(vi) $((V,E_t),\{ (u_i,v_i) \mid i\in[k],\{v_i,u_i\} \not=\{z_i\},v_i \not=s_i\})$ is a \yes-instance of \textsc{Disjoint Path}.
		By assumption and (iii), there are temporally 
		non-intersecting temporal walks $P_1,\dots,P_k$ 
		such that 
		(a) $P_i$ is a temporal $(s_i,u_i)$-walk in $\TG$ with arrival time $t_i\leq t-1$, and
		(b) $\forall i,j \in [k],u_i\not=u_j\colon \exists p\in\{i,j\}\colon u_p=s_p \vee (u_p=z_p \wedge t_p<t-1)$.
		Let $K_s:=\{ i \in [k] \mid s_i=v_i\}$, $K_z:=\{ i \in [k] \mid u_i=z_i=v_i\}$, and $K=[k]\setminus(K_s \cup K_z)$.	
		For all $i \in K_s$ we set $P_i' = (\emptyset)$.
		For all $i \in K_z$ we set $P_i' = P_i$.
		From (vi), there are vertex-disjoint paths $P_1^2,\dots,P_{|K|}^2$ 
		such that $P_i^2$ is a $(u_i,v_i)$-path in $(V,E_t)$, for all $i\in K$.
		For all $i \in K$ we construct a temporal $(s_i,v_i)$-walk $P_i'$ 
		by concatenating $P_i$ and $P_i^2$.
		Note that $P_1',\dots,P_k'$ are temporally non-intersecting  such that
		(a) $P_i'$ is a temporal  $(s_i,v_i)$-walk with arrival time $t_i'\leq t$, for all $i\in [k]$.
		Moreover, if we have for some $i,j\in [k]$ that $v_i=v_j$, 
		then there is a $p \in \{i,j\}$ such that either $v_p = s_p$ or $v_p = z_p$, because of (i).
		If $v_p = z_p$, then $u_p=z_p$ and thus $t_i'<t$, because of (ii).
		Hence, (b) $\forall i,j \in [k],i\not=j\colon v_i\not=v_j \vee v_i=s_i \vee (v_i=z_i \wedge t_i'<t)$ is true.

		So, $I$ is a \yes-instance if and only if $D[T,z_1,\dots,z_k] = \top$.

		\paragraph{Running time.}
		The table size of $D$ is $O(|V|^k\cdot\lifetime)$.
		To compute one entry of $D$ we look at up to $O(|V|^k)$ other entries of~$D$ and 
		then solve an instance of \textsc{Disjoint Paths} in $O(n^2)$ time\footnote{This
		running time bound considers $k$ to be a constant. There is a factor $f(k)$
		for some function $f$ hidden in the $O$-notation.}
		\cite{KAWARABAYASHI2012424}.
		This gives an overall running time of $O(|V|^{2k+2}\cdot\lifetime)$ if $k$ is
		a fixed constant.
\end{proof}

Finally, we point out that \cref{thm:NP-h-for-3-path} implies that for \TempDisjointPaths{} we presumably
cannot achieve a result similar to \cref{thm:xp} while \cref{thm:w1}
implies that we presumably cannot improve this result on a classification level,
that is we cannot expect to find an FPT-algorithm for \WalkVersion\ for
parameter $|S|$.

\section{Temporal lines and trees}\label{sec:unterlyingpath}
In this section, we investigate the computational complexity of \AllProbs\ for
restricted classes of underlying graphs, in particular so-called \emph{temporal
lines} and \emph{temporal trees}. The former are temporal graphs that have a
path as underlying graph and the latter are temporal graphs that have a tree as
underlying graph.
In particular, we first show that,
surprisingly, the problems remain NP-hard on temporal lines (and thus also on temporal trees). 
On the positive side we show that, on temporal trees, 
the \PathVersion{} is fixed-parameter tractable with respect to the
number of source-sink pairs. 
The latter result stands in stark contrast to the general case, 
where the problem is NP-hard even when the number of source-sink pairs is two~(\cref{thm:NP-h-for-3-path}). If we further restrict all source-sink pairs
to consist of the two end-points of the temporal line, then we obtain a
polynomial-time algorithm.

Before we proceed with our results in this section, first we recall some background on foremost temporal paths. 
Given a temporal graph $\TG$ and two specific vertices $s,z$ of it, a \emph{foremost} temporal path from $s$ to $z$ starting at time $t$ is 
a temporal path which starts at $s$ not earlier than at time $t$ and arrives at $z$ with the earliest possible arrival time. 
A foremost temporal path from $s$ to~$z$ starting at time~$t$ can be computed in linear $O\left(|V|+\sum_{i=1}^{\lifetime}|E_i|\right)$ time~\cite{WuCKHHW16}.

\begin{restatable}{theorem}{thmnpcpaths}%
		\label{thm:NPcompleteOnPaths}
\AllProbs\ is \NP-hard even on a temporal line where all temporal paths are to the same direction.
\end{restatable}
\begin{proof}%
We present here a polynomial-time reduction for 
\STempDisjointWalks. 
The reduction is done from \textsc{Multicolored independent
set on unit interval graphs}, which is known to be
NP-complete~\cite[Lemma~2]{Bevern2015}.
In this problem the input is a unit interval graph $G=(V,E)$ with $n$ vertices,
where $V$~is partitioned into $k$ subsets of independent vertices; 
we interpret each of these subsets as a vertex color.
The goal is to
compute an independent set of size $k$ in $G$ which contains exactly one vertex from
each color.
By possibly slightly shifting the endpoints of the intervals in the given unit
interval representation of $G$, we can assume without loss of generality that
all endpoints of the intervals are distinct.
Furthermore, we can assume without loss of generality that each interval
endpoint is an integer between $k+1$ and $k+n^2$ (while all intervals still
have the same length).

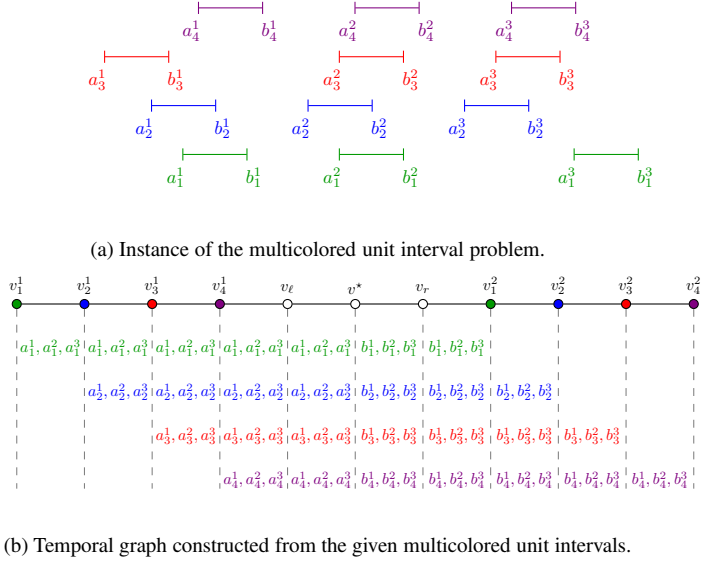
\begin{figure*}[t!]
\centering
\begin{subfigure}{.42\textwidth}
\centering
\begin{tikzpicture}[scale = 0.65, every node/.style={scale=0.7},xscale=.8]
\tikzstyle{vertlin}=[inner sep=0.1,draw,minimum size=.9mm]
\def\gr{green!60!black}
\def\bl{blue}
\def\rd{red}
\def\ye{violet}
\def\x{20} %
\def\y{1} %

\node[minimum height=4.8cm] at (1,1.2) {}; %

\node[label=below:{\color{\gr}$a_1^1$}] (A11) at ($\x*(0.1,0)$)   {};
\node[label=below:{\color{\gr}$b_1^1$}] (B11) at ($ (A11) + \x*(0.1,0)$)  {};
\draw[|-|,color=\gr] (A11) -- (B11);

\node[label=below:{\color{\gr}$a_1^2$}] (A12) at ($(B11)  + \x*(0.1,0)$)  {};
\node[label=below:{\color{\gr}$b_1^2$}] (B12) at ($(A12) + \x*(0.1,0)$)  {};
\draw[|-|,color=\gr] (A12) -- (B12);

\node[label=below:{\color{\gr}$a_1^3$}] (A13) at ($(B12)  + 2*\x*(0.1,0)$)  {};
\node[label=below:{\color{\gr}$b_1^3$}] (B13) at ($(A13) + \x*(0.1,0)$)  {};
\draw[|-|,color=\gr] (A13) -- (B13);

\node[label=below:{\color{\bl}$a_2^1$}] (A21) at ($\x*(0.06,0) + \y*(0,1)$)  {};
\node[label=below:{\color{\bl}$b_2^1$}] (B21) at ($(A21) + \x*(0.1,0)$)  {};
\draw[|-|,color=\bl] (A21) -- (B21);

\node[label=below:{\color{\bl}$a_2^2$}] (A22) at ($(B21)  + \x*(0.1,0)$)  {};
\node[label=below:{\color{\bl}$b_2^2$}] (B22) at ($(A22) + \x*(0.1,0)$)  {};
\draw[|-|,color=\bl] (A22) -- (B22);

\node[label=below:{\color{\bl}$a_2^3$}] (A23) at ($(B22)  + \x*(0.1,0)$)  {};
\node[label=below:{\color{\bl}$b_2^3$}] (B23) at ($(A23) + \x*(0.1,0)$)  {};
\draw[|-|,color=\bl] (A23) -- (B23);

\node[label=below:{\color{\rd}$a_3^1$}] (A31) at ($(0,0) + 2*\y*(0,1)$)  {};
\node[label=below:{\color{\rd}$b_3^1$}] (B31) at ($(A31) + \x*(0.1,0)$)  {};
\draw[|-|,color=\rd] (A31) -- (B31);

\node[label=below:{\color{\rd}$a_3^2$}] (A32) at ($(B31)  + 2*\x*(0.1,0)$)  {};
\node[label=below:{\color{\rd}$b_3^2$}] (B32) at ($(A32) + \x*(0.1,0)$)  {};
\draw[|-|,color=\rd] (A32) -- (B32);

\node[label=below:{\color{\rd}$a_3^3$}] (A33) at ($(B32)  + \x*(0.1,0)$)  {};
\node[label=below:{\color{\rd}$b_3^3$}] (B33) at ($(A33) + \x*(0.1,0)$)  {};
\draw[|-|,color=\rd] (A33) -- (B33);

\node[label=below:{\color{\ye}$a_4^1$}] (A41) at ($\x*(.12,0) + 3*\y*(0,1)$)  {};
\node[label=below:{\color{\ye}$b_4^1$}] (B41) at ($(A41) + \x*(0.1,0)$)  {};
\draw[|-|,color=\ye] (A41) -- (B41);

\node[label=below:{\color{\ye}$a_4^2$}] (A42) at ($(B41)  + \x*(0.1,0)$)  {};
\node[label=below:{\color{\ye}$b_4^2$}] (B42) at ($(A42) + \x*(0.1,0)$)  {};
\draw[|-|,color=\ye] (A42) -- (B42);

\node[label=below:{\color{\ye}$a_4^3$}] (A43) at ($(B42)  + \x*(0.1,0)$)  {};
\node[label=below:{\color{\ye}$b_4^3$}] (B43) at ($(A43) + \x*(0.1,0)$)  {};
\draw[|-|,color=\ye] (A43) -- (B43);

\end{tikzpicture}
\caption{Instance of the multicolored unit interval problem.}
\end{subfigure}\\
\begin{subfigure}{.56\textwidth}
\centering
\begin{tikzpicture}[scale = 0.45, every node/.style={scale=0.6},yscale=1.3]
\def\gr{green!60!black}
\def\bl{blue}
\def\rd{red}
\def\ye{violet}
\def\ln{gray}
\def\x{2} %
\def\y{-1} %

\node[minimum height=5.6cm] at (2,-2) {}; %

\node[vert2,label=above:$v_1^1$, fill = \gr] (V11) at (1,0) {};
\node[vert2,label=above:$v_2^1$, fill = \bl] (V21) at ($ (V11) + \x*(1,0)$){};
\node[vert2,label=above:$v_3^1$, fill = \rd] (V31) at ($ (V21) + \x*(1,0)$) {};
\node[vert2,label=above:$v_4^1$, fill = \ye] (V41) at ($ (V31) + \x*(1,0)$) {};
\node[vert2,label=above:$v_\ell$] (Vell) at ($ (V41) + \x*(1,0)$) {};
\node[vert2,label=above:$v^\star$] (Vstar) at ($ (Vell) + \x*(1,0)$) {};
\node[vert2,label=above:$v_r$] (Vr) at ($ (Vstar) + \x*(1,0)$) {};
\node[vert2,label=above:$v_1^2$,fill= \gr] (V12) at ($ (Vr) + \x*(1,0)$){};
\node[vert2,label=above:$v_2^2$, fill = \bl] (V22) at ($ (V12) + \x*(1,0)$) {};
\node[vert2,label=above:$v_3^2$, fill = \rd] (V32) at ($ (V22) + \x*(1,0)$) {};
\node[vert2,label=above:$v_4^2$, fill = \ye] (V42) at ($ (V32) + \x*(1,0)$) {};

\path[draw] (V11) --(V21);
\node at ($ (V11) + 0.5*\x*(1,0) + \y*(0,1)$) {\color{\gr} $a_1^1, a_1^2, a_1^3$};

\path[draw] (V21) --(V31);
\node at ($ (V21) + 0.5*\x*(1,0) + \y*(0,1)$) {\color{\gr} $a_1^1, a_1^2, a_1^3$};
\node at ($ (V21) + 0.5*\x*(1,0) + 2*\y*(0,1)$) {\color{\bl} $a_2^1, a_2^2, a_2^3$};

\path[draw] (V31) --(V41);
\node at ($ (V31) + 0.5*\x*(1,0) + \y*(0,1)$) {\color{\gr} $a_1^1, a_1^2, a_1^3$};
\node at ($ (V31) + 0.5*\x*(1,0) + 2*\y*(0,1)$) {\color{\bl} $a_2^1, a_2^2, a_2^3$};
\node at ($ (V31) + 0.5*\x*(1,0) + 3*\y*(0,1)$) {\color{\rd} $a_3^1, a_3^2, a_3^3$};

\path[draw] (V41) --(Vell);
\node at ($ (V41) + 0.5*\x*(1,0) + \y*(0,1)$) {\color{\gr} $a_1^1, a_1^2, a_1^3$};
\node at ($ (V41) + 0.5*\x*(1,0) + 2*\y*(0,1)$) {\color{\bl} $a_2^1, a_2^2, a_2^3$};
\node at ($ (V41) + 0.5*\x*(1,0) + 3*\y*(0,1)$) {\color{\rd} $a_3^1, a_3^2, a_3^3$};
\node at ($ (V41) + 0.5*\x*(1,0) + 4*\y*(0,1)$) {\color{\ye} $a_4^1, a_4^2, a_4^3$};

\path[draw] (Vell) --(Vstar);
\node at ($ (Vell) + 0.5*\x*(1,0) + \y*(0,1)$) {\color{\gr} $a_1^1, a_1^2, a_1^3$};
\node at ($ (Vell) + 0.5*\x*(1,0) + 2*\y*(0,1)$) {\color{\bl} $a_2^1, a_2^2, a_2^3$};
\node at ($ (Vell) + 0.5*\x*(1,0) + 3*\y*(0,1)$) {\color{\rd} $a_3^1, a_3^2, a_3^3$};
\node at ($ (Vell) + 0.5*\x*(1,0) + 4*\y*(0,1)$) {\color{\ye} $a_4^1, a_4^2, a_4^3$};

\path[draw] (Vstar) --(Vr);
\node at ($ (Vstar) + 0.5*\x*(1,0) + \y*(0,1)$) {\color{\gr} $b_1^1, b_1^2, b_1^3$};
\node at ($ (Vstar) + 0.5*\x*(1,0) + 2*\y*(0,1)$) {\color{\bl} $b_2^1, b_2^2, b_2^3$};
\node at ($ (Vstar) + 0.5*\x*(1,0) + 3*\y*(0,1)$) {\color{\rd} $b_3^1, b_3^2, b_3^3$};
\node at ($ (Vstar) + 0.5*\x*(1,0) + 4*\y*(0,1)$) {\color{\ye} $b_4^1, b_4^2, b_4^3$};

\path[draw] (Vr) -- (V12);
\node at ($ (Vr) + 0.5*\x*(1,0) + \y*(0,1)$) {\color{\gr} $b_1^1, b_1^2, b_1^3$};
\node at ($ (Vr) + 0.5*\x*(1,0) + 2*\y*(0,1)$) {\color{\bl} $b_2^1, b_2^2, b_2^3$};
\node at ($ (Vr) + 0.5*\x*(1,0) + 3*\y*(0,1)$) {\color{\rd} $b_3^1, b_3^2, b_3^3$};
\node at ($ (Vr) + 0.5*\x*(1,0) + 4*\y*(0,1)$) {\color{\ye} $b_4^1, b_4^2, b_4^3$};

\path[draw] (V12) --(V22);
\node at ($ (V12) + 0.5*\x*(1,0) + 2*\y*(0,1)$) {\color{\bl} $b_2^1, b_2^2, b_2^3$};
\node at ($ (V12) + 0.5*\x*(1,0) + 3*\y*(0,1)$) {\color{\rd} $b_3^1, b_3^2, b_3^3$};
\node at ($ (V12) + 0.5*\x*(1,0) + 4*\y*(0,1)$) {\color{\ye} $b_4^1, b_4^2, b_4^3$};

\path[draw] (V22) --(V32);
\node at ($ (V22) + 0.5*\x*(1,0) + 3*\y*(0,1)$) {\color{\rd} $b_3^1, b_3^2, b_3^3$};
\node at ($ (V22) + 0.5*\x*(1,0) + 4*\y*(0,1)$) {\color{\ye} $b_4^1, b_4^2, b_4^3$};

\path[draw] (V32) --(V42);
\node at ($ (V32) + 0.5*\x*(1,0) + 4*\y*(0,1)$) {\color{\ye} $b_4^1, b_4^2, b_4^3$};

\draw [dashed,color=\ln] (V11) -- ($(V11) + 4.2*\y*(0,1)$);
\draw [dashed,color=\ln] (V21) -- ($(V21) + 4.2*\y*(0,1)$);
\draw [dashed,color=\ln] (V31) -- ($(V31) + 4.2*\y*(0,1)$);
\draw [dashed,color=\ln] (V41) -- ($(V41) + 4.2*\y*(0,1)$);
\draw [dashed,color=\ln] (Vell) -- ($(Vell) + 4.2*\y*(0,1)$);
\draw [dashed,color=\ln] (Vstar) -- ($(Vstar) + 4.2*\y*(0,1)$);
\draw [dashed,color=\ln] (Vr) -- ($(Vr) + 4.2*\y*(0,1)$);
\draw [dashed,color=\ln] (V12) -- ($(V12) + 4.2*\y*(0,1)$);
\draw [dashed,color=\ln] (V22) -- ($(V22) + 4.2*\y*(0,1)$);
\draw [dashed,color=\ln] (V32) -- ($(V32) + 4.2*\y*(0,1)$);
\draw [dashed,color=\ln] (V42) -- ($(V42) + 4.2*\y*(0,1)$);

\end{tikzpicture}
\caption{Temporal graph constructed from the given multicolored unit intervals.}
\end{subfigure}
\caption{Example of the reduction described in the proof of \cref{thm:NPcompleteOnPaths}.}
\label{fig:reductionIntMCIS}
\end{figure*}

\paragraph{Construction.} From the given multi-colored unit intervals in $G$,
we construct a temporal line $\TP$ using the following procedure.
Let $\{c_1, \dots, c_k\}$ be a set of all colors of the intervals in $G$.
First we fix an arbitrary linear ordering $c_1 < c_2 < \ldots < c_k$ of the $k$
colors, and we add to the underlying path $P$ of $\TP$ two vertices $v_i^1$ and
$v_i^2$, for every color $c_i$.
We add to $P$ also three basis vertices $v_\ell , v^\star, v_r$.
The vertices of~$P$ are ordered starting from $v_1^1,v_2^1,\ldots,v_k^1$, followed
by the basis vertices $v_\ell, v^\star, v_r$, and finishing with
$v_1^2,v_2^2,\ldots,v_k^2$.
At the end we have $P = (v_1^1, v_2^1, \dots, v_k^1, v_\ell, v^\star, v_r,
v_1^2, v_2^2, \dots, v_k^2)$.

We construct the multiset~$S$ of source-sink pairs as follows. 
Let~$m_i$ be the number of intervals of color $c_i$. 
For every color~$c_i$ we add the pair $(v_i^1, v_i^2)$ to $S$. We refer to this
source-sink pair as ``the verification source-sink pair for color~$c_i$''.
Furthermore, we add $m_i-1$ copies of the pair $(v_i^1, v_\ell)$ to $S$ and we add $m_i-1$
copies of pair $(v_r, v_i^2)$ to $S$. We call these $2m_i-2$ source-sink
pairs the ``dummy source-sink pairs for color $c_i$''. 

To fully define the
temporal line~$\TP$, we still need to add time labels to the edges of $P$.
Denote by $a_i^j$ and $b_i^j$ the start and end points of the $j$th interval of
color $c_i$.
We set up the edge labels of the path $P$ from $v_i^1$ to $v_i^2$ as follows.
To edge $\{v_{s}^1,  v_{s+1}^1\}$ with $s\in[k-1]$ we add the labels~$a_i^j$
with $i\le s$. To edges $\{v_{k}^1,  v_\ell\}$ and $\{v_\ell,  v^\star\}$ we add
all labels~$a_i^j$.
To edge $\{v_{s}^2,  v_{s+1}^2\}$ with $s\in[k-1]$ we add the labels~$b_i^j$
with $i> s$. To edges $\{v^\star,  v_r\}$ and $\{v_r,  v^2_1\}$ we
add all labels~$b_i^j$.
See \cref{fig:reductionIntMCIS} for an example. The construction can clearly be
performed in polynomial time.

\paragraph{Correctness.} ($\Rightarrow$): 
Assume there is a multicolored independent set $V'\subseteq V$ in $G$. Let
$v_i\in V'$ be the vertex in the independent set with color $c_i$ and let
$[a_i^j,b_i^j]$ be the interval of $v_i$. Then for the verification source-sink
pair of $c_i$ we use the following temporal path:
$((v^1_i, v^1_{i+1}, a_i^j), (v^1_{i+1}, v^1_{i+2}, a_i^j), \ldots, (v^1_{k-1},
v^1_{k}, a_i^j),$ $(v^1_k, v_\ell, a_i^j), (v_\ell, v^\star, a_i^j), (v^\star,
v^r, b_i^j), (v_r, v^2_{1}, b_i^j), (v^2_1, v^2_2, b_i^j),$ $\ldots, (v^2_{i-1},
v^2_i, b_i^j))$. For the dummy source-sink pairs $(v_i^1, v_\ell)$ of~$c_i$ we use the temporal paths $((v^1_i, v^1_{i+1}, a_i^{j'}),\ldots,$ 
$(v^1_{k-1}, v^1_{k}, a_i^{j'}), (v^1_k, v_\ell, a_i^{j'}))$ with $j'\neq j$.
Note that there are exactly $m_i-1$ pairwise different paths of this kind. Analogously, for the
dummy source-sink pairs $(v_r, v_i^2)$ of~$c_i$ we use the temporal paths
$((v_r, v^2_{1}, b_i^{j'}), (v^2_1, v^2_2, b_i^{j'}),\ldots,$ $(v^2_{i-1},
v^2_i, b_i^{j'}))$ with $j'\neq j$. It is easy to check that the temporal paths
for the dummy source-sink pairs of all colors do not temporally intersect. Now
assume for contradiction that two verification source-sink pairs of colors $c_i$
and $c_{i'}$ temporally intersect. Then they have to intersect in $v^\star$,
since this is the only vertex where the paths wait. By construction, the
temporal path for the verification source-sink pairs of color
$c_i$ visits $v^\star$ during the interval $[a_i^j,b_i^j]$ and the verification source-sink pairs of color
$c_{i'}$ visits $v^\star$ during the interval $[a_{i'}^{j'},b_{i'}^{j'}]$. These
two intervals correspond to the intervals of the vertices of colors $c_i$ and
$c_{i'}$ in the multicolored independent set $V'$. Hence, those intervals
intersecting is a contradiction to the assumption that $V'$ is in fact an
independent set.

($\Leftarrow$): 
Assume we have a set of pairwise temporally disjoint walks for all source-sink
pairs in $S$. Note that all edges except $\{v_\ell, v^\star \}$
and $\{v^\star, v_r\}$ have as many time labels as temporal walks that need to
go through them. Furthermore, note that $\{v_\ell, v^\star \}$ has the same
labels as $\{v^1_k, v_\ell \}$ and $\{v^\star, v_r\}$ has the same labels as
$\{v_r, v^2_1\}$. This in particular
implies that all temporal walks are in fact paths since the only vertex that
could be visited by a path for more than one time step is $v^\star$.
Therefore, for every pair
$(s,z)\in S$, no temporal path from $s$ to~$z$ can ever stop and wait at any vertex different from~$v^\star$.
Furthermore, the only paths going through vertex~$v^\star$ are the paths
connecting vertices $v_i^1$ and $v_i^2$ (which correspond to color $c_i$); we will refer
to this path as the color path of $c_i$.
Consider color $c_1$ and its dummy source-sink pairs $(v_1^1, v_\ell)$. By
construction, the edge $\{v_1^1,v_2^1\}$ has time labels corresponding to the
start points $a_1^j$ of intervals from the $m_1$~vertices of~$G$ that have color
$c_1$. It follows that the temporal paths for these dummy source-sink pairs and
the color path of $c_1$ use only time labels corresponding to the
start points $a_1^j$ of intervals from the $m_1$ vertices of $G$ that have color
$c_1$ until they are at $v_\ell$ or arrive at $v^\star$, respectively, since
they cannot wait at any vertex. Now by induction, this holds for all other
colors $c_i$ and by an analogous argument, this also holds for the ``second
half''. More specifically, we also have that temporal paths for the dummy
source-sink pairs $(v_r, v_i^2)$ as well as the ``second part'' of the color
path of $c_i$ use time labels corresponding to end points $b_i^j$ of intervals from the vertices of~$G$ that have color
$c_i$ when going from $v_r$ (respectively $v^\star$) to their corresponding
destinations.

It follows that each color path can enter and leave vertex~$v^\star$ only at the
time corresponding to the start and end points of its color intervals.
In any other case some of the other vertices are blocked, which prevents the
completion of others temporal $S$-paths.
Recall that intervals of the same color are non-overlapping. Hence, for every
color path corresponding to a color $c_i$ we can find one interval $[a_i^j,b_i^j]$ 
such that the color path visits $v^\star$ in an interval that includes
$[a_i^j,b_i^j]$. Since the color paths are temporally non-intersecting, the
vertices corresponding to the intervals form a multicolored independent set in
$G$.
\end{proof}

Next, we show fixed-parameter tractability of \PathVersion{}
parameterized by the number $|S|$ of source-sink pairs if the
underlying graph is a tree.

\begin{restatable}{theorem}{thmfpt}%
		\label{thm:fpt}
\PathVersion{} on temporal trees is in FPT when parameterized
by $|S|$, as it can be solved in $O\left(|S|^{|S|+3}\cdot |\TG|\right)$ time.
\end{restatable}
\begin{proof}%
Let $I = (\TG, S)$ be an instance of \PathVersion{}, 
the underlying graph~$G$ %
being a tree and $S$~consists of 
$k$~source-sink pairs $(s_1,z_1), \ldots, (s_k,z_k)$.
We solve~$I$ using the following procedure.

First we can observe that, since $G$ is a tree, every source-sink pair
$(s_i,z_i)$ in $S$ corresponds to exactly one path $P_i$ in $G$.
Furthermore, if two paths $P_1$ and $P_2$ intersect in a tree, their
intersection $P_1 \cap P_2$ is a continuous path in $G$ (potentially containing only one vertex).
In the case that $P_1, P_2$ intersect, there are two ways that $P_1 \cap P_2$ can be traversed: 
either first by $P_1$ and then $P_2$, or vice-versa. 
The main idea of our algorithm is that we enumerate all possible ways to traverse the intersections of each pair of paths corresponding to two source-sink pairs.

Note that the set of $k$ paths in the tree $G$ (where every path corresponds
to one source-sink pair) have the Helly property~\cite{bollobas1986combinatorics}.
That is, for any three of these paths $P_i, P_j, P_\ell$, if they pairwise have a non-empty intersection then 
$P_i \cap P_j \cap P_\ell \neq \emptyset$. 
Therefore, in order to enumerate all possible ways to traverse the intersections of each pair of source-sink paths, it suffices to just 
enumerate all permutations of these $k$ paths. 
We do this as follows.
Let $V = \{v_1, v_2, \dots , v_k \}$ be the set of vertices, where each $v_i$ corresponds to the path $P_i$ of the source-sink pair $(s_i, z_i) \in S$,
and let $\pi$ be an arbitrary permutation of $V$. 
Let $i \neq j \in [k]$. If $v_i <_\pi v_j$ then we route the path $P_i$ before the path $P_j$ on their intersection $P_1 \cap P_2$ 
(assuming $P_1 \cap P_2 \neq \emptyset$). 
What remains to check is whether $\pi$ leads to a feasible routing of the $k$ source-sink pairs paths. 
We check this efficiently using an auxiliary (static) directed graph $G_\pi$ on the vertices of $V$, defined as follows. 
Any two vertices $v_i$ and $v_j$ are connected with an oriented arc in $G_\pi$ if and only if $P_i \cap P_j \neq \emptyset$. 
In the latter case, the arc is oriented from $v_i$ to $v_j$ if and only if $v_i <_\pi v_j$.

It is not hard to see that $G_\pi$ has no directed cycles, and thus there is at least one vertex with in-degree $0$.
Let $U$ be the set of all such vertices.
Clearly, since the vertices of $U$ have no incoming arcs, their corresponding paths in $G$ are pairwise non-intersecting. 
Even more, each of these paths $P_i$ (where $v_i \in U$) visits each of its edges before any other path $P_j$. 
Therefore we determine the temporal edges of every such path $P_i$ by computing a foremost temporal path 
(i.e.,~a temporal path with the earliest arrival time) from $s_i$ to $z_i$ in~$\TG$. 
Let $v_x$ be an arbitrary internal vertex of one of these temporal paths $\mathcal{P}$, 
and suppose that $v_x$ is visited by $\mathcal{P}$ within the time interval $[a_x,b_x]$. 
Then we remove from the edges $\{v_{x-1},v_x\}, \{v_x,v_{x+1}\}$ of
$\mathcal{P}$ all labels $\ell\leq b_x$.
Furthermore, we also remove the vertex set $U$ from $G_\pi$ and repeat the
procedure, until either all $k$ temporal paths $P_1, \ldots, P_k$ are routed according to the permutation $\pi$ 
or some path $P_i$ cannot be routed. 
In the latter case we select the next (not yet investigated) permutation $\pi'$ and repeat the whole procedure.

During the above procedure we construct $O(k!)=O(k^{k+1})$ different $k$-vertex graphs $G_\pi$ (one for each permutation).
The construction of each one of them takes $O(k^2)$ time. 
For each auxiliary graph $G_\pi$ we can sequentially calculate all foremost temporal 
paths in linear $O(|V|+\sum_{i=1}^{\lifetime}|E_i|)$ time~\cite{WuCKHHW16}.
In total, since $|S|=k$, all computations can be done in
$O\left(|S|^{|S|+3}\cdot \left(|V|+\sum_{i=1}^{\lifetime}|E_i|\right)\right)$
time.
\end{proof}

We remark that it remains open whether a similar result can be obtained for
\WalkVersion, since we cannot assume w.l.o.g.\ that the
temporally disjoint walks are actually paths, even on temporal lines (for
an example see \cref{fig:exampleWalks}).
Presumably (and in contrast to the general case) the walk version is
computationally more difficult than the path version of our problem on temporal
paths and trees.

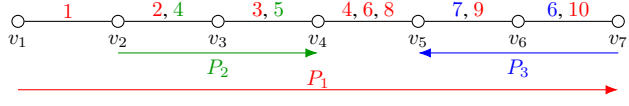
\begin{figure}
\centering
\begin{tikzpicture}[scale = 0.7, every node/.style={scale=0.8},xscale=.95]
\def\x{2} %
\def\y{0.2} %
\def\gr{green!60!black}
\def\bl{blue}
\def\rd{red}

\node[vert2,label=below:$v_1$] (V1) at (1,0) {};
\node[vert2,label=below:$v_2$] (V2) at ($ (V1) + \x*(1,0)$) {};
\node[vert2,label=below:$v_3$] (V3) at ($ (V2) + \x*(1,0)$) {};
\node[vert2,label=below:$v_4$] (V4) at ($ (V3) + \x*(1,0)$) {};
\node[vert2,label=below:$v_5$] (V5) at ($ (V4) + \x*(1,0)$) {};
\node[vert2,label=below:$v_6$] (V6) at ($ (V5) + \x*(1,0)$) {};
\node[vert2,label=below:$v_7$] (V7) at ($ (V6) + \x*(1,0)$) {};

\path[draw] (V1) --(V2);
\node at ($ (V1) + 0.5*\x*(1,0) + \y*(0,1)$) {{\color{\rd}$1$}};
\path[draw] (V2) --(V3);
\node at ($ (V2) + 0.5*\x*(1,0) + \y*(0,1)$) {{\color{\rd}$2$}, {\color{\gr}$4$}};
\path[draw] (V3) --(V4);
\node at ($ (V3) + 0.5*\x*(1,0) + \y*(0,1)$) {{\color{\rd}$3$}, {\color{\gr}$5$}};
\path[draw] (V4) --(V5);
\node at ($ (V4) + 0.5*\x*(1,0) + \y*(0,1)$) {{\color{\rd}$4$}, {\color{\rd}$6$}, {\color{\rd}$8$}};

\path[draw] (V5) --(V6);
\node at ($ (V5) + 0.5*\x*(1,0) + \y*(0,1)$) {{\color{\bl}$7$}, {\color{\rd}$9$}};
\path[draw] (V6) --(V7);
\node at ($ (V6) + 0.5*\x*(1,0) + \y*(0,1)$) {{\color{\bl}$6$}, {\color{\rd}$10$}};

\draw[-Latex,color=\rd] ($(V1) - (0,1.3)$) -- ($ (V7) - (0,1.3)$);
\node at ($(V4) -(0,1.1)$) {\color{\rd}$P_1$} {};
\draw[-Latex,color=\gr] ($(V2) - (0,0.6)$) -- ($ (V4) - (0,0.6)$);
\node at ($(V3) -(0,0.9)$) {\color{\gr}$P_2$} {};
\draw[-Latex,color=\bl] ($(V7) - (0,0.6)$) -- ($ (V5) - (0,0.6)$);
\node at ($(V6) -(0,0.9)$) {\color{\bl}$P_3$} {};

\end{tikzpicture}
\caption{Temporally disjoint walks on a temporal line are not necessarily equal
to temporally disjoint paths. 
Suppose that one wants to determine edge-labels of the following walks: $P_1$ from $v_1$ to $v_7$, $P_2$ from $v_2$ to $v_4$ and $P_3$ from $v_7$ to $v_5$ on the depicted temporal graph. 
In the feasible solution $P_2$ and $P_3$ are temporal paths, but $P_1$ has to be a walk.}
\label{fig:exampleWalks}
\end{figure}

Finally, we show that we can solve \AllProbs\ in polynomial time if the
underlying graph is a path and all source-sink pairs consist of the endpoints of
that path.

\begin{restatable}{theorem}{thmpoly}%
		\label{thm:poly}
Let $\mathcal{G}$ be a temporal line having $P = (v_0, v_1, v_2, \ldots, v_n)$ as its underlying path. 
If $S$ contains $k$ times the source-sink pair $(v_0, v_n)$ and $\ell=|S|-k$ times the source-sink pair $(v_n, v_0)$, 
then \AllProbs\ can be solved on $\mathcal{G}$ in polynomial time, 
namely $O\left(k\ell\left(k+\ell\right)\cdot |\TG|\right)$.
\end{restatable}
\begin{proof}%
We present here the proof for the problem version \STempDisjointPaths. 
Let $I = (\TP, S)$ be an instance of \STempDisjointPaths, 
where $\TP$ is a given temporal line  %
with $P = (v_0, v_1, v_2, \ldots, v_n)$ as its underlying path. 
Assume that there have to be $k$ (resp.~$\ell=|S|-k$) temporally disjoint $(v_0,v_n)$- (resp.~$(v_n,v_0)$-) paths in the output, \ie they must have the orientation from $v_0$ to~$v_n$,  (resp.~from $v_n$ to $v_0$).

We solve the instance $I$ using dynamic programming. The main idea is that, since all temporal paths
start and end in endpoints of $P$,
in any optimal solution, once a temporal path starts, it proceeds in the fastest possible way 
(without interfering with previously started paths). 
Therefore, assuming we start with $(v_0, v_n)$-temporal paths, we only need to find out \emph{how many} 
$(v_0, v_n)$-temporal paths follow the starting path, after that \emph{how many} $(v_n,v_0)$-temporal paths follow, then after that \emph{how many} $(v_0,v_n)$-temporal paths follow, etc.

Let ${0\leq i \leq k}$, ${0\leq j \leq \ell}$, and ${1\leq t \leq \lifetime}$. 
Then $L(i,j,t)$ denotes the earliest arrival time of
$(k-i) + (\ell-j)$ temporally non-intersecting temporal paths with 
$k-i$ being $(v_0, v_n)$-temporal paths and $\ell-j$ being $(v_n,v_0)$-temporal paths,
assuming that the earliest-starting temporal path is a $(v_0,v_n)$-temporal path that starts at time~$t$.
If it is not possible to route such $(k-i)+(\ell-j)$ temporally non-intersecting temporal paths starting at time $t$, then let $L(i,j,t) = \infty$. 
Similarly we define $R(i,j,t)$,
with the only difference that here the earliest-starting 
temporal path needs to start at time~$t$ 
from~$v_n$ and finishes at~$v_0$. 
For the sake of completeness, we let $L(i,j,\infty)=R(i,j,\infty) = \infty$ for every $i\leq k$ and every $j\leq \ell$. 
Furthermore, for every $t$, every $i\leq k-1$, and every $j\leq \ell-1$, we let $L(k,j,t) = R(i,\ell,t) = \infty$. 
Finally we let $L(k,\ell,t) = R(k,\ell,t) = t-1$. 
Note that, the input instance $I$ is a \yes-instance if and only if $\min \{
L(0,0,1),\ R(0,0,1)\} \neq \infty$.
Furthermore, note that, for every triple $i,j,t$, the value $\min\{L(i,j,t),R(i,j,t)\}$ is the earliest arrival time of all temporal paths 
in the subproblem where, until time $t-1$, exactly $i$ and $j$ temporally non-intersecting temporal $(v_0,v_n)$- and $(v_n,v_0)$-paths, respectively, have been routed.

The value $L(i,j,t)$ can be recursively computed as follows. 
Suppose that, in the optimal solution, ${1\leq p\leq k-i}$ temporally non-intersecting $(v_0,v_n)$-temporal paths are first routed (starting at time $t$) before the first $(v_n,v_0)$-temporal path (among the $\ell-j$ ones) is routed.
Let $t_p$ be the earliest arrival time of these $p$ paths if they can all be routed; if not, then we set $t_p = \infty$.
Then:
\begin{equation}
\label{L-eq}
L(i,j,t) = \min \{R(i+p,j,t_{p}+1) \mid 1\leq p\leq k-i\}.
\end{equation}
The value $R(i,j,t)$ can be computed similarly:
\begin{equation}
\label{R-eq}
R(i,j,t) = \min \{L(i,j+p,t^{*}_{p}+1) \mid 1\leq p\leq \ell-j\}, 
\end{equation}
where $(v_n,v_0)$-temporal paths are routed.

The values $\{t_p \mid 1\leq p \leq k-i\}$ can be computed as follows. 
If $p=1$ then $t_p$ is the arrival time of the $(v_0,v_n)$-foremost temporal path~$P_1$. %
To determine $t_2$, we first compute %
$P_1$ and %
then, for every internal vertex $v_x$ of $\TP$, if $v_x$ is visited by $P_1$ within the time interval $[a_x,b_x]$, 
we remove from the edges $\{v_{x-1},v_x\}, \{v_x,v_{x+1}\}$ of $\TP$ all labels $l\leq b_x$. In the resulting temporal line we then compute 
the foremost temporal path~$P_2$, %
which arrives at $v_n$ at time $t_2$. 
By applying this procedure iteratively, we either compute $p$ temporally non-intersecting 
temporal paths $P_1, P_2, \ldots, P_p$, 
starting at time $t$ and arriving at time $t_p$,
or we conclude that $t_p=\infty$. 
The values $\{t^{*}_p \mid 1\leq p \leq \ell-j\}$ (for the $(v_n, v_0)$-temporal paths) can be computed in a symmetric way. 
All these computations together can be done in linear time.

From the above it follows that
we can decide \STempDisjointPaths\ by checking whether $\min\{L(0,0,1),R(0,0,1)\}$ is finite or not. 
In total, there are $2k\ell T$ values $L(i,j,t)$ and $R(i,j,t)$. 
Observe that, for every pair $i,j$, we only need to compute the value $L(i,j,t)$  (resp.~$R(i,j,t')$) for one specific value of $t$ (resp.~$t'$).
This observation ensures that the running time of the algorithm is polynomial.
For this we need the following observation. Assume that, in the recursive tree originated at $L(0,0,1)$, 
we need to compute at two different places the values $L(i,j,t)$ and $L(i,j,t')$, where $t' > t$. Then, since obviously $L(i,j,t) \leq L(i,j,t')$, 
at the second place we can just replace $L(i,j,t')$ by $\infty$, thus stopping the recursive calculations at that branch of the recursive tree. 
Similarly, if we need to compute at two different places the values $R(i,j,t)$ and $R(i,j,t')$, where $t' > t$, 
we replace $R(i,j,t')$ by $\infty$ at the second place of the recursive tree. That is, for every pair $i,j$, we just need to compute only one value $L(i,j,t)$ (resp.~$R(i,j,t)$). Therefore, we can build two matrices $M^L$ and $M^R$, each of size $(k+1)\times (\ell+1)$, such that 
$M^L (i,j)$ (resp.~$M^R (i,j)$) stores the unique value of $t$ for which we need to compute $L(i,j,t)$ (resp.~$R(i,j,t)$). 
That is, in the recursion tree originated at $L(0,0,1)$, for every pair $i,j$ we only need to compute the values $L(i,j,M^L (i,j))$ 
and $R(i,j,M^R (i,j))$.

Similarly, for the recursive tree originated at $R(0,0,1)$ we need to build two other matrices $N^L$ and $N^R$ (each of size $(k+1)\times (\ell+1)$) 
for the same purpose, as the recursive tree originated at $R(0,0,1)$ is different to the one originated at $L(0,0,1)$. 
That is, in the recursion tree originated at $R(0,0,1)$, for every pair $i,j$ we only need to compute the values $L(i,j,N^L (i,j))$ 
and $R(i,j,N^R (i,j))$. 

Each of these four $(k+1)\times (\ell+1)$ matrices can be computed by running $O(k\ell(k+\ell))$ times the foremost temporal path algorithm 
(in order to compute at each step in linear time $O\left(|V|+\sum_{i=1}^{\lifetime}|E_i|\right)$ the values $\{t_p \mid 1\leq p \leq k-i\}$ and $\{t^{*}_p \mid 1\leq p \leq \ell-j\}$, respectively). 
Once we have built these four matrices, we can iteratively compute the value $L(0,0,1)$ (resp.~$R(0,0,1)$) 
in at most $k \ell$ computations, each of which takes at most $O(k+\ell)$ time (cf.~equations (\ref{L-eq})-(\ref{R-eq})). 
Therefore, all computations can be done in $O\left(k\ell\left(k+\ell\right)\cdot \left(|V|+\sum_{i=1}^{\lifetime}|E_i|\right)\right)$ time.

This completes the proof for the case of \TempDisjointPaths. 
Finally, it is easy to see that in the problem~\TempDisjointWalks, 
in any optimal solution every temporal walk is a temporal path, as every temporal walk is from $v_0$ to $v_n$ or from $v_n$ to $v_0$. 
Therefore, the above algorithm for \TempDisjointPaths\ also solves
\TempDisjointWalks. 
\end{proof}

\section{Conclusion}
Introducing temporally disjoint paths and walks, we modeled
the property that agents moving along these 
never meet, even though they might visit the same vertices. 
We identified an unexpected difference in their computational complexity: 
\STempDisjointPaths\ is NP-hard even for two paths, while 
\STempDisjointWalks\ can be done in polynomial time for a constant number of walks 
(however it becomes W[1]-hard when parameterized by the number of walks). 
On the contrary, while \STempDisjointPaths\ becomes fixed-parameter tractable for the number of
paths if the underlying graph is a path,
we leave open whether we can obtain a similar result for \STempDisjointWalks\, which seems to be  much more complicated than the path version.
\bibliographystyle{named}
\bibliography{bibliography}

\end{document}